\documentclass[%
 reprint,
 amsmath,amssymb,
 aps,
]{revtex4-1}

\usepackage{amsthm}
\usepackage{graphicx}% Include figure files
\usepackage{dcolumn}% Align table columns on decimal point
\usepackage{bm}% bold math
\usepackage{amsthm}
\begin{document}

\preprint{APS/123-QED}

\title{Strong quantum nonlocality for multipartite entangled states}% Force line breaks with \\

\author{Zhi-Chao Zhang$^{1}$}
  \email{zhichao858@126.com}
\author{Guojing Tian$^{2}$}
  \email{tianguojing@ict.ac.cn}
  \author{Tian-Qing Cao$^{3}$}
 % \email{henancaotianqing@163.com}

\affiliation{%
 $^{1}$School of Mathematics and Physics, University of Science and Technology Beijing, Beijing, 100083, China\\
$^{2}$Institute of Computing Technology, Chinese Academy of Sciences, 100190, Beijing, China\\
$^{3}$School of Mathematical Sciences, Tiangong University, Tianjin, 300387, China\\}

\date{\today}% It is always \today, today,

\begin{abstract}
Recently, Halder \emph{et al.} [S. Halder \emph{et al.}, Phys. Rev. Lett. \textbf{122}, 040403 (2019)] present two sets of strong nonlocality of orthogonal product states based on the local irreducibility. However, for a set of locally indistinguishable orthogonal entangled states, the remaining question is whether the states can reveal strong quantum nonlocality. Here we present a general definition of strong quantum nonlocality based on the local indistinguishability. Then, in $2 \otimes 2 \otimes 2$ quantum system, we show that a set of orthogonal entangled states is locally reducible but locally indistinguishable in all bipartitions, which means the states have strong nonlocality. Furthermore, we generalize the result in N-qubit quantum system with $N\geqslant 3$. Finally, we also construct a class of strong nonlocality of orthogonal entangled states in $d\otimes d\otimes \cdots \otimes d, d\geqslant 3$. Our results extend the concept of strong nonlocality for entangled states.
\begin{description}
\item[PACS numbers]
03.67.Hk, 03.65.Ud, 03.67.Mn
\end{description}
\end{abstract}

\pacs{Valid PACS appear here}% PACS, the Physics and Astronomy
\maketitle

%\tableofcontents

\section{\label{sec:level1}Introduction\protect}
Quantum nonlocality is one of the most important properties in quantum information theory, and the most well-known manifestation of quantum nonlocality is Bell nonlocality [1], which means a quantum state violates Bell-type inequalities [2-5]. Apart from the foundational implications, Bell nonlocality has many applications in quantum technologies [6-8]. 

In fact, nonlocal properties also have another class which is different from Bell nonlocality. Specifically, when a set of orthogonal quantum states cannot be perfectly distinguished by local operations and classical communication (LOCC), it reflects the fundamental feature of quantum mechanics which is called nonlocality [9]. Locally distinguishing orthogonal quantum states refers to that, in a quantum system which consists of several parts held by separated observers, a state secretly chosen from a set of prespecified orthogonal quantum states is shared by these parties, and the goal is to determine the shared state, only using LOCC [10-21]. In addition, the nonlocality of orthogonal quantum states can be used for practical purposes, such as, data hiding [22,23], quantum secret sharing [24], and so on. Thus, in the past two decades, it has received considerable attention to study the locally distinguishability of orthogonal quantum states and explore the relationship between quantum nonlocality and quantum entanglement [25-39]. 

All the above results prompt our researchers to make further study on more extensive quantum nonlocality. Recently, in $d \otimes d \otimes d, d=3,4$, Halder \emph{et al.} proved that two sets of orthogonal product states are strongly nonlocal because these states are locally irreducible in all bipartitions [40]. Here, local irreducibility means that it is impossible to locally eliminate one or more states from the set of states while preserving orthogonality of the post measurement states. Then, the authors in [41] present the general definition of strong quantum nonlocality based on the local irreducibility.  However, the local irreducibility is only a sufficient but not necessary condition for local indistinguishability. Furthermore, for orthogonal product states, the authors in [42,43] present the classification based on the local distinguishability when all the parties are spatially separated or different subsets of the parties come together, and design the local state discrimination protocols for such states with additional entangled recources. 

Thus, the above results leave the following questions: (i) for entangled states, how to define the property of strong quantum nonlocality based on the local indistinguishability, (ii) whether we can find strongly nonlocal orthogonal entangled states, and how many states can have strong nonlocality.

Motivated by the above questions, in this work, we first present the general definition of strong quantum nonlocality based on the local inditinguishability, and our definition is more general than the definition in [40] because the local irreducibility is only a sufficient but not necessary condition for local indistinguishability. Then, in $2 \otimes 2 \otimes 2$, we show one set of strong nonlocality of orthogonal entangled states to explain the above definition, that is, these states are locally reducible but locally indistinguishable. Furthermore, in N-qubit quantum system, where $N\geqslant 3$, we generalize the above quantum states and prove that these states are strongly nonlocal. Finally, we also construct a class of strong nonlocality of entangled states in $d\otimes d\otimes \cdots \otimes d, d\geqslant 3$. In addition, our results can also let people have a better understanding of the relationship between entanglement and nonlocality. 

The rest of this paper is organized as follows. In Sec. \uppercase\expandafter{\romannumeral2}, we present a general definition of strong quantum nonlocality. Then, we show some sets of orthogonal entangled states are strongly nonlocal in N-qubit quantum system, where $N\geqslant 3$. Furthermore, in $d\otimes d\otimes \cdots \otimes d, d\geqslant 3$, we construct a class of strong nonlocality of entangled states in Sec. \uppercase\expandafter{\romannumeral3}. Finally, we summarize in Sec. \uppercase\expandafter{\romannumeral4}.

\theoremstyle{remark}
\newtheorem{definition}{\indent Definition}
\newtheorem{lemma}{\indent Lemma}
\newtheorem{theorem}{\indent Theorem}
\newtheorem{corollary}{\indent Corollary}

\def\QEDclosed{\mbox{\rule[0pt]{1.3ex}{1.3ex}}}
\def\QED{\QEDclosed}
\def\proof{\indent{\em Proof}.}
\def\endproof{\hspace*{\fill}~\QED\par\endtrivlist\unskip}

\section{The general definition of strong quantum nonlocality}

In this section, we present a general definition of strong quantum nonlocality based on the local indistinguishability of orthogonal quantum states. In Ref.[40], for $d_{1}\otimes d_{2}\otimes \cdots \otimes d_{n}$ quantum system, the authors have defined the local irreducibility of a set of orthogonal quantum states, which means that it is not possible to locally eliminate one or more states from the set while preserving orthogonality of the post measurement states. Then, based on the local irreducibility, the authors in [40, 41] defined the strong quantum nonlocality. From the definition of local irreducibility, we know that a set of locally irreducible orthogonal quantum states must be locally indistinguishable, so these states have nonlocality, but the converse does not always hold. Thus, we think the best appropriate definition of strong nonlocality should use the locally indistinguishability. 

In the following, we present a general definition for strong quantum nonlocality based on the locally indistinguishability. 

\begin{definition}
In $d_{1}\otimes d_{2}\otimes \cdots \otimes d_{n}, n\geqslant 3$, a set of orthogonal quantum states is strongly nonlocal if it is locally indistinguishable in every bipartition, where bipartition means the whole quantum system is divided into two parts. 
\end{definition}

From the Definition 1, we can know that it is different from the definition of strong nonlocality in [40] and more general. In [40], the authors present that the three-qubit GHZ basis (unnormalized): $|000\rangle \pm |111\rangle$, $|011\rangle \pm |100\rangle$, $|001\rangle \pm |110\rangle$, $|010\rangle \pm |101\rangle$ is locally reducible in every bipartition. Nevertheless, we find that these states are still locally indistinguishable in every bipartition. According to our definition, these states are strongly nonlocal. Thus, our definition is more general. In the following, we present that even some of three-qubit GHZ basis are strongly nonlocal. First, we present the following states in $2\otimes 2\otimes 2$:

\begin{eqnarray}
\label{eq.2}
\begin{split}
&|\phi_{1,2}\rangle=|000\rangle \pm |111\rangle, \\
&|\phi_{3}\rangle=|011\rangle + |100\rangle,\\
&|\phi_{4}\rangle=|001\rangle + |110\rangle, \\
&|\phi_{5}\rangle=|010\rangle + |101\rangle.\\
\end{split}
\end{eqnarray}

Then, we prove the above states are locally indistinguishable in every bipartition.

\begin{theorem}
In $2\otimes 2\otimes 2$,  the above $5$ states are strongly nonlocal.
\end{theorem}

\begin{proof}
First, we consider the bipartition $AB|C$. Physically this means that the subsystems $A$ and $B$ are treated together as a four-dimensional subsystem $AB$ on which joint measurements are now allowed. To reflect this, we rewrite the states $|\phi_{1,2,4}\rangle$ as follows:

\begin{eqnarray}
\label{eq.2}
\begin{split}
&|\phi_{1,2}\rangle=|0_2 0\rangle \pm |1_21\rangle, \\
&|\phi_{4}\rangle=|0_21\rangle + |1_20\rangle.\\
\end{split}
\end{eqnarray}
where $|0_2\rangle$ means the first $|00\rangle$, and $|1_2\rangle$ means first $|11\rangle$.

Then, the states (2) cannot be locally distinguished across $AB|C$ because the above states are locally equivalent to three Bell states $|00\rangle \pm |11\rangle$, $|01\rangle + |10\rangle$ which are locally indistinguishable [37, 40]. Thus, the states (1) cannot be locally distinguished across $AB|C$. 

From the structure of states (1), we know these states have a cyclic property as the cyclic property of the trace. Then,  the states (1) are also locally indistinguishable across $B|CA$ and $C|AB$. 

Therefore, the states (1) are strongly nonlocal.
This completes
the proof.          
\end{proof}

In the following, we generalize the above result in N-qubit quantum systems. In the same way, we first present some of N-qubit GHZ basis states as follows.

\begin{eqnarray}
\label{eq.2}
\begin{split}
&|\phi_{1}\rangle=|00\ldots0\rangle - |11\ldots1\rangle, \\
&|\phi_{2+a_{2}2^{N-2}+a_{3}2^{N-3}+\ldots+a_{N}2^0}\rangle=|a_{1}\ldots a_{N}\rangle + |\bar{a_{1}}\ldots \bar{a_{N}}\rangle.\\
\end{split}
\end{eqnarray}
where $a_{1}=0$, $a_{2},\ldots, a_{N}=0$ or $1$, $\bar{a_{i}}=(a_{i}+1)\bmod 2, i=1,\ldots, N$.

Then, we have the following result.

\begin{theorem}
In N-qubit quantum systems $2\otimes 2\otimes \ldots \otimes2$,  the above $2^{N-1}+1$ states of (3) are strongly nonlocal.
\end{theorem}

\begin{proof}
First, we consider any bipartition $A_{1} \ldots A_{j}|A_{j+1} \ldots A_{N}$. Physically this means that the subsystems $A_{1}, \ldots, A_{j}$ are treated together as a $2^{j}$-dimensional subsystem $A_{1} \ldots A_{j}$ on which joint measurements are now allowed. Others parties are similar. Here $A_{1}, \ldots, A_{N}$ can be any parties. Then, there must be one state in which the first $j$ parties are all in state $|0\rangle$ (or $|1\rangle$) and simultaneously the last $N-j$ parties are all in state $|1 \rangle$ (or $|0\rangle$), thus we can rewrite the state as $|\phi \rangle=|0_{j}\rangle |1_{N-j}\rangle +|1_{j} \rangle |0_{N-j}\rangle$ under new basis. Similar to Theorem 1, we consider the state $|\phi \rangle$ and $|\phi_{1,2}\rangle$ which are locally equivalent to three Bell states $|00\rangle \pm |11\rangle$, $|01\rangle + |10\rangle$. Thus, the states (3) cannot be locally distinguished by LOCC in the bipartition $A_{1} \ldots A_{j}|A_{j+1} \ldots A_{N}$.

Therefore, the states (3) are strongly nonlocal.
This completes
the proof.          
\end{proof}

In addition, the states of (3) are locally reducible in every bipartition. However, we have proved that these states are locally indistinguishable and strongly nonlocal. Thus, our definition should be more general and suitable. In [40], the authors show that such product states cannot be found in systems where one of subsystems has dimension two. However, for entangled states, the minimum dimension can be two according to the above results. In the following, we construct strongly nonlocal maximally entangled states (MESs) in more general quantum systems.

\section{Strong nonlocality of orthogonal maximally entangled states}

In this section, we will present the explicit set of strongly nonlocal MESs in $d \otimes d \otimes d$ and $d \otimes d \otimes \cdots \otimes d$ quantum systems respectively, where $d \geqslant 3$.

\subsection{Strongly nonlocal MESs in tripartite quantum systems}
To clearly explain the general strong nonlocality in tripartite quantum systems, we need start with the construction of strongly nonlocal MESs in $3 \otimes 3 \otimes 3$ quantum system. 

\begin{lemma}
	In $3 \otimes 3 \otimes 3$ quantum system, the following $6$ MESs
	\begin{eqnarray}
	|000\rangle + |111\rangle + |222\rangle,\nonumber \\
	|000\rangle + \omega |111\rangle + \omega^2 |222\rangle, \nonumber \\
	|000\rangle + \omega^2 |111\rangle + \omega |222\rangle,\nonumber \\
	|100\rangle + |211\rangle + |022\rangle,\nonumber \\
	|010\rangle + |121\rangle + |202\rangle,\nonumber \\
	|001\rangle + |112\rangle + |220\rangle. \nonumber 
	\end{eqnarray} 
	are strongly nonlocal, where $\omega=e^{\frac{2\pi i}{3}}$.
\end{lemma}

\begin{proof}
From the definition of strong nonlocality, if we prove the nonlocality in every bipartite separation, i.e., $A|BC$, $B|AC$ and $C|AB$,  then the above $6$ MESs have strong nonlocality.

In $A|BC$ separation, we set last two $|00\rangle$ as $|0_2\rangle$, $|11\rangle$ as $|1_2\rangle$, $|22\rangle$ as $|2_2\rangle$, thus the first four states of the above set can be rewritten as 
	\begin{eqnarray}
	|00_2\rangle + |11_2\rangle + |22_2\rangle,\nonumber \\
	|00_2\rangle + \omega |11_2\rangle + \omega^2 |22_2\rangle, \nonumber \\
	|00_2\rangle + \omega^2 |11_2\rangle + \omega |22_2\rangle,\nonumber \\
	|10_2\rangle + |21_2\rangle + |02_2\rangle,\nonumber
	\end{eqnarray} 
These four states can be regarded as in a new $3 \otimes 3$ quantum system, in which the computational basis is $\{ |00_2\rangle, |01_2\rangle, |02_2\rangle, |10_2\rangle, |11_2\rangle, 12_2\rangle, |20_2\rangle, |21_2\rangle, |22_2\rangle \}$. Based on the result that any $d+1$ MESs cannot be distinguished by LOCC in [18], the set of these four states has quantum nonlocality in this $A|BC$ separation.

Similarly, we can prove the set including the first three and the fifth states cannot be locally distinguished in $B|AC$ separation, and the set including the first three and the sixth states cannot be locally distinguished in $C|AB$ separation.
\end{proof}

It is not hard to construct strongly nonlocal MESs sets in a most general tripartite quantum system according to the idea in the above proof. Thus we can derive the following theorem.

\begin{theorem}
	In $d \otimes d \otimes d$ quantum system, where $d \geqslant 3$, the following $d+3$ MESs
	\begin{eqnarray}
	\sum_{l=0}^{d-1} \omega^{jl} |lll\rangle, j=0,1,\cdots, d-1, \nonumber \\	
	|100\rangle + |211\rangle + \cdots + |0,d-1,d-1\rangle,\nonumber \\
	|010\rangle + |121\rangle + \cdots + |d-1,0,d-1\rangle,\nonumber \\
	|001\rangle + |112\rangle + \cdots + |d-1,d-1,0\rangle, \nonumber 
	\end{eqnarray} 
	are strongly nonlocal, where $\omega=e^{\frac{2\pi i}{d}}$.
\end{theorem}

Actually in $d \otimes d$ quantum system, any $d+1$ MESs have already been used to represent quantum nonlocality [18]. Here from this theorem, we know the fact that adding just $2$ more MESs to the nonlocal MESs can realize the strong nonlocality in $d \otimes d \otimes d$ quantum system.

\subsection{Strongly nonlocal MESs in more than tripartite quantum systems}

When the number of quantum subsystems is bigger than $3$, the construction of strongly nonlocal MESs will become a little different. In this subsection, we will present the explicit form of strongly nonlocal MESs in more than tripartite quantum systems.

\begin{theorem}
	In a $k$-partite quantum system $d \otimes d \otimes \cdots \otimes d$, where $k \geqslant 4, d \geqslant 3$, the following $(k+1)d$ MESs
	\begin{eqnarray}
	\sum_{l=0}^{d-1} \omega^{jl} |ll \cdots l\rangle, j=0,1,\cdots, d-1, \nonumber \\	
	\sum_{l=0}^{d-1} \omega^{jl} |l\oplus_d1,l, \cdots, l\rangle, j=0,1,\cdots, d-1, \nonumber \\	
	\sum_{l=0}^{d-1} \omega^{jl} |l,l\oplus_d1, \cdots, l\rangle, j=0,1,\cdots, d-1, \nonumber \\	
	\vdots \nonumber \\
	\sum_{l=0}^{d-1} \omega^{jl} |l,l, \cdots, l\oplus_d1\rangle, j=0,1,\cdots, d-1, \nonumber
	\end{eqnarray} 
	are strongly nonlocal, where $\omega=e^{\frac{2\pi i}{d}}$.
\end{theorem}

\begin{proof}
	To prove the strong nonlocality of these $(k+1)d$ MESs, we need to prove they are nonlocal in every separation. Next we will complete this proof case by case.
	
	In ``1|(k-1)'' separation, without loss of generality, we will take $P_1 | P_2 \cdots P_{k}$ as an example. The first $d+1$ MESs in the above set can be rewritten as 
	\begin{eqnarray}
	\sum_{l=0}^{d-1} \omega^{jl} |ll_{k-1}\rangle, j=0,1,\cdots, d-1, \nonumber \\	
	\sum_{l=0}^{d-1}  |l\oplus_d1,l_{k-1}\rangle, \nonumber \\	
	\end{eqnarray} 
	where $| l_{k-1}\rangle$ denotes $| l \rangle ^{\otimes (k-1)}$. Then these $d+1$ MESs can be regarded as in another new $d \otimes d$ quantum system where the computational basis of the second subsystem is $\{ | l_{k-1} \rangle \}_{l=0}^{d-1}$. Thus these MESs have nonlocality in $P_1 | P_2 \cdots P_{k}$ separation. Similarly, we can prove quantum nonlocality in other ``1|(k-1)'' separations.
	
	In ``2|(k-2)'' separation, the number of new basis states increases. Taking the $P_1 P_2 | P_3 \cdots P_{k}$ separation as an example, the new computational basis for ``$P_1 P_2$'' is $\{ |jj\rangle, |j \oplus_d 1,j \rangle, |j, j\oplus_d 1 \} _{j=0}^{d-1}$, that is, ``$P_1 P_2$'' can be regarded as a new $3d$-dimension subsystem. Meanwhile, the latter subsystem  ``$P_3 \cdots P_{k}$'' can be regarded as a new $(k-1)d$-dimension subsystem because its computational basis is $\{ |jj \cdots j \rangle, |j\oplus_d 1,j,\cdots, j \rangle, |j,j,\cdots, j\oplus_d 1 \rangle \}_{j=0}^{d-1}$. Thus the new dimension for $P_1 P_2 | P_3 \cdots P_{k}$ separation is $3d \otimes (k-1)d$, where $k \geqslant 4$. Then we need $(k-1)d+1$ MESs to illustrate the quantum nonlocality in this separation, which can be ensured by the fact that we have $(k+1)d$ MESs in the original set. Other ``2|(k-2)'' separations can be similarly proved.
	
	Next, we will consider the ``3|(k-3)'' separation. Actually, the number of new basis states is no more than `2|(k-2)'' separation, or even exact the symmetric case of the former cases, so the original $(k+1)d $ MESs can also assure the quantum nonlocality in every ``3|(k-3)'' separation. 
	
	The following ``m|(k-m)'' separation with $m \geqslant 4$ will become easier, or even exact the symmetric case of the former cases. Until now, we have proved the quantum nonlocality in every possible separations, which exactly satisfies the definition of quantum strong nonlocality.	
\end{proof}

\section{Conclusion and Discussion}

In summary, we have shown the definition for strong nonlocality of orthogonal quantum states and constructed some sets of strongly nonlocal orthogonal quantum states in $d\otimes d\otimes \ldots \otimes d$, thus extending the concept of strong nonlocality. Our results can lead to a better understanding of the relationship between nonlocality and entanglement. In addition, for orthogonal product states, there are some locally distinguishing protocols with entanglement resource [42-47]. However, for entangled states, it is very few, so it is interesting to investigate the less entanglement resource required to distinguish entangled states especially the above entangled states.

However, for a more than tripartite quantum system, the definition of nonlocality is not complete. For example, in $d_{A}\otimes d_{B}\otimes d_{C} \otimes d_{D}$ quantum system, (i) when a set of orthogonal quantum states is locally indistinguishable in a 4-partition $A|B|C|D$, we know that these states have nonlocality that we presently understand; (ii) when a set of orthogonal quantum states is locally indistinguishable in every bipartition (e.g., $AB|CD$ and $ABC|D$), these states have strong nonlocality, such as our results in above section; but (iii) when a set of orthogonal quantum states is locally indistinguishable in every tripartition (e.g., $AB|C|D$ and $BD|A|C$), but locally distinguishable in some bipartition, the nonlocality of these states should be defined. 

As [41, 42], we also can classify the different strength of nonlocality of orthogonal quantum states based on the local indistinguishability. Here, we use $\mathcal{N}$ to indicate the strength of nonlocality of a set of orthogonal quantum states and get the relationship as follows.

\begin{eqnarray}
\label{eq.2}
\begin{split}
\mathcal{N}_{2}>\cdots>\mathcal{N}_{i+}>\mathcal{N}_{i}>\cdots\\
>\mathcal{N}_{n+}>\mathcal{N}_{n}, i=3, \cdots, n-1,
\end{split}
\end{eqnarray}
where $\mathcal{N}_{j}$, $j=2,\cdots ,n$, means that a set of orthogonal quantum states is only locally indistinguishabel in every $j$-partition and $\mathcal{N}_{j+}$, $j=3,\cdots ,n$, means that a set of orthogonal quantum states is locally indistinguishable in every $j$-partition and also locally indistinguishable in only some $(j-1)$-partition.

Therefore, from the above relationship, we can present the definition for strong nonlocality. In the following, super-LOCC means that there are at least two parties treated together as a subsystem on which joint measurements are allowed, and the $n$-parties are at least divided into $2$ parts.
\begin{definition}
In an $n$-partite quantum system, where $n>2$, a set of orthogonal quantum states is strongly nonlocal if it cannot be perfectly distinguished by super-LOCC.
\end{definition}

From the Definition 2, we know that super-LOCC is more powerful than LOCC, but less powerful than global operations. Thus the definition should be more general and appropriate.

In addition, we recently find that the authors in Ref. [48] have also presented a class of strong nonlocality of entangled states based on the local irreducibility in $d \otimes d \otimes d, d\geqslant 3$, and in their construction, strong nonlocality needs almost $d^3$ entangled states. But in our construction method, because we use local indistinguishability, it only needs $d+3$ entangled states to show our new-defined strong nonlocality, and generalize the construction in $d\otimes d\otimes \ldots \otimes d, d\geqslant 2$.

\begin{acknowledgments}
The authors are grateful for the anonymous referees' suggestions to improve the quality of this paper. This work was supported by the Beijing Natural Science Foundation (Grant No. 4194088), the NSFC (Grants No. 61901030, No. 61801459, No. 61701343 and No. 11847210), the National Postdoctoral Program for Innovative Talent (Grant No. BX20180042 and No. BX20190335), the China Postdoctoral Science Foundation (Grant No. 2018M640070), and the Anhui Initiative in Quantum Information Technologies (Grant No. AHY150100).
\end{acknowledgments}

\nocite{*}

\bibliography{apssamp}% Produces the bibliography via BibTeX.

\end{document}